\documentclass[letterpaper, 10 pt, conference]{ieeeconf}  
\IEEEoverridecommandlockouts                              

\usepackage{xcolor} 
\usepackage{amssymb} 
\usepackage{amsmath} 
\usepackage{bigints} 
\usepackage{braket} 
\usepackage{physics} 
\usepackage{graphicx} 
\usepackage{url}
\newtheorem{theorem}{Theorem}

\newtheorem{lemma}{Lemma}

\newtheorem{definition}{Definition}

\newtheorem{remark}{Remark}

\usepackage{tikz-cd}
\usepackage{tikz} 
\usetikzlibrary{chains, positioning, shapes.symbols,shapes,arrows}
\definecolor{darkblue}{RGB}{9,72,90}
\definecolor{lightblue}{RGB}{124,184,201}
\usepackage{pgf} 

\usepackage{enumerate} 
\usepackage{caption}
\usepackage{tikz-cd}
\usepackage{tikz} 
\usetikzlibrary{chains, positioning, shapes.symbols,shapes,arrows}
\definecolor{darkblue}{RGB}{9,72,90}
\definecolor{lightblue}{RGB}{124,184,201}
\usepackage{pgf} 

\definecolor{amethyst}{HTML}{9966CC} 
\definecolor{amber}{HTML}{FFBF00}    
\definecolor{aqua}{HTML}{00FFFF}     

\title{ \bf
    Dynamic Programming Principle and Stabilization for Mean-Field Quantum Filtering Systems}

\author{Sofiane Chalal$^{1}$, Nina H. Amini$^{2}$, Hamed Amini$^{3}$, Mathieu Laurière$^{4}$ 
\thanks{$^{1}$ S. Chalal is with L2S, Universit\'{e} Paris-Saclay,
        CentraleSupélec, France
        {\tt\tiny sofiane.chalal@centralesupelec.fr}}%
\thanks{$^{2}$ N.H. Amini is with CNRS, L2S, Universit\'{e} Paris-Saclay,
        CentraleSupélec, France
        {\tt\tiny nina.amini@centralesupelec.fr}}%
\thanks{$^{1}$ H. Amini is with Department of Industrial and Systems Engineering, University of Florida, Gainesville, FL, USA.
{\tt\tiny aminil@ufl.edu.}}
\thanks{$^{4}$ M. Laurière is with Shanghai Frontiers Science Center of AI and DL NYU-ECNU Institute of Mathematical Sciences NYU Shanghai
        {\tt\tiny mathieu.lauriere@nyu.edu}.}}
            
\makeatletter 
\let\NAT@parse\undefined
\makeatother
\usepackage[colorlinks = true,
            linkcolor = blue,
            urlcolor  = blue,
            citecolor = blue,
            anchorcolor = blue]{hyperref} 

\begin{document}

\maketitle
\thispagestyle{empty}
\pagestyle{empty}


\begin{abstract}
Working within the quantum filtering framework, we establish a dynamic programming principle in an infinite-dimensional setting by embedding the state space into the Hilbert–Schmidt space. We then study a stabilization problem for continuously monitored Ising-coupled qubits and, in the mean-field limit, demonstrate quantum state reduction together with exponential convergence toward prescribed eigenstates under suitable feedback laws.
\end{abstract}

\section{Introduction }
\medskip

Quantum control is essential for the development of emerging quantum technologies.
In practice, quantum devices operate as \emph{open} systems: they interact with an environment,
are subject to decoherence, and are probed through imperfect measurements.
A particularly powerful paradigm is \emph{continuous-time measurement and feedback},
where an experimenter steers the dynamics in real time using the measurement record.
In this setting, the controller does not access the underlying physical state directly;
instead, decisions are based on the observation process and on the associated
\emph{conditional state} produced by quantum filtering \cite{belavkin09qdpp,belavkin87non}.
This viewpoint places quantum feedback control in close analogy with classical
stochastic control under partial observation \cite{boutenhandel07,boutenhandel08,handel05}.

In the classical theory, stochastic control provides a canonical framework for synthesizing
feedback policies for dynamical systems driven by noise and incomplete information.
When several decision-makers act simultaneously, the natural extension is that of
\emph{differential games}, which generalize stochastic control to multi-agent settings:
each player selects its own control according to its own objective, while the state evolves
under the combined actions of all agents.
When the number of interacting agents becomes large, however, direct analysis and computation
quickly become intractable. This motivates \emph{mean-field} formulations, where weak interactions
and symmetry allow one to approximate the collective effect through a limiting 
dynamics. Mean-field control and mean-field games have therefore become standard tools for
modeling and solving large-scale control and game-theoretic problems \cite{bensoussan13mean}.

A similar hierarchy of ideas has emerged in the quantum regime.
Kolokoltsov~\cite{kolokoltsov22dynamic} first introduced a continuous-time framework for dynamic quantum games.
Subsequently, \cite{kolokoltsov21law} derived a mean-field equation for $N$-particle quantum systems
under continuous observation and proposed corresponding classes of quantum mean-field games~\cite{kolokoltsov22qmfg,kolokoltsov21qmfgcounting}.
In the finite-dimensional case $\mathbb{H}=\mathbb{C}^d$, taking the coupling operators to be the Gell--Mann matrices leads to a non-degenerate diffusion on the projective space $\mathbf{P}(\mathbb{C}^d)$ with generator equal to the Laplace--Beltrami operator. This geometric structure facilitates the analysis of the corresponding Hamilton--Jacobi equations and permits the application of classical stochastic control methods on manifolds. In contrast, our goal is to formulate quantum control and quantum
dynamic games \emph{in full generality} by working directly with density operators, which is essential
for infinite-dimensional models and for settings where such geometric parametrizations are not available.

A cornerstone of classical stochastic control is the \emph{dynamic programming principle} (DPP),
which yields the Hamilton--Jacobi--Bellman (HJB) equation and characterizes the value function through a partial differential equation. 
In the quantum filtering framework, an analogous route is conceptually available:
the conditional density operator is a sufficient statistic for feedback control,
and the controlled conditional dynamics is a  stochastic differential equation
with values in the space of density operators \cite{belavkin09qdpp,belavkin09cybernetics,boutenhandel07}. However, in infinite dimension the natural state space is the Banach space of trace-class operators,
where standard It\^o calculus is not directly applicable. This creates a technical obstruction to establishing a rigorous DPP/HJB theory beyond
finite-dimensional matrix models.

The first objective of this paper is to overcome this obstruction by working in the
Hilbert--Schmidt space.
Since trace-class operators embed into the Hilbert--Schmidt space,
we can view density operators as elements of a Hilbert space and invoke stochastic calculus
in a well-posed setting.
We formalize Fr\'echet differentiability on the Hilbert--Schmidt space,
identify the associated second-order operator induced by the controlled filtering dynamics,
and prove a DPP for the value function of a broad class of running and terminal costs.
As a consequence, we derive the corresponding HJB equation in an infinite-dimensional operator setting.
This provides a rigorous dynamic-programming foundation for quantum feedback control problems
whose underlying Hilbert space is infinite-dimensional.

In the second part, we investigate the stabilization problem for a system of continuously monitored $N$-qubits with Ising-type interactions, formulated as a dynamic game. By passing to the mean-field limit $N \to \infty$, we prove quantum state reduction and exponential stabilization toward a prescribed equilibrium state of the limiting dynamics. The result is obtained under an appropriate continuous-time linear feedback control, using a fixed-point argument to characterize and stabilize the mean-field behavior.

\paragraph*{Organization} The remainder of the paper is organized as follows.
Section~\ref{sec:notations} introduces notations and recalls basic facts from quantum theory and operator spaces.
In Section~\ref{sec:formalism}, we review quantum filtering and feedback control theroy; we then establish the dynamic programming principle in infinite dimension and  derive the associated HJB equation.
Section~\ref{sec:nbody} develops the $N$-body continuously monitored dynamics framework, and discusses the mean-field limit leading to the limiting mean-field  quantum filtering equation.
In Section~\ref{sec:game}, we formulate the problem of quantum state preparation for $N$-body monitored Ising-coupled qubits. 
Passing to the mean-field limit, we prove quantum state reduction and exponential stabilization under feedback, and we present a two-qubit game example illustrating the game-theoretic interpretation and the implementation of local feedback via reduced states.
Finally, Section~\ref{sec:conc} concludes the paper and outlines perspectives for future work.


\section{Notations}
\label{sec:notations}
The sets of real and complex numbers are denoted by $\mathbb{R}$ and $\mathbb{C}$ respectively , $\mathcal{M}_d(\mathbb{C})$ denotes the set of complex $d\times d$ matrices. 
For any $z \in \mathbb{C}$, 
$\overline{z}$ denotes its complex conjugate, and the roman symbol $\mathrm{i}$ is used 
for the imaginary unit. For any integer $N \ge 1$, we use the notation $[N] := \{1,\dots,N\}$. We denote a measurable space $(\mathfrak{X},\mathcal{X})$ endowed with a measure $\mu$. Fix a finite time horizon $T > 0$, and $\mathcal{U}$ denotes the set of admissible controls which is the set of  $\mathbb{R}$-valued bounded progressively measurable process $(\alpha_{t})_{0 \leq t \leq T}$ with respect to the underlying filtration, which will be clear from context.

We fix the notation $\mathbb{H}$ for a complex separable Hilbert space endowed with the scalar product $\langle \bullet,\bullet\rangle : \mathbb{H}\times \mathbb{H} \to \mathbb{C}$, with associated norm $\|\psi \| := \sqrt{\langle \psi,\psi\rangle}$.  The set $\mathcal{B}(\mathbb{H})$ denotes the set of bounded linear operator on $\mathbb{H}$. For every $ {O} \in \mathcal{B}(\mathbb{H})$, denote by ${{O}}^{\dagger}$ its adjoint operator. We denote $\mathrm{tr} : \mathcal{B}(\mathbb{H}) \mapsto \mathbb{C}$ trace of any operator, the space $\mathcal{B}_{1}(\mathbb{H})$ (resp. $\mathcal{B}_{2}(\mathbb{H})$) is the space of trace-class operators (resp. of Hilbert-Schmidt operators):
{\small\begin{align*}
    \varrho \in \mathcal{B}_1(\mathbb{H}) &\Longleftrightarrow \mathrm{tr}(|\varrho|) < \infty, \quad
    \varrho \in \mathcal{B}_2(\mathbb{H}) \Longleftrightarrow \mathrm{tr}(\varrho^{\dagger}\varrho) < \infty,
\end{align*}
}

Endowed with inner product 
{\small \begin{align*}
 \langle \varrho_1, \varrho_2 \rangle_{2} := \sum_{n \in \mathbb{N}} \langle \varrho_1\psi_n,\varrho_2\psi_n\rangle = \mathrm{tr}(\varrho_1^{\dagger}\varrho_{2}),
\end{align*}}
for all $\varrho_1,\varrho_{2} \in \mathcal{B}_{2}(\mathbb{H})$, the induced norm $\|\varrho\|_{2} := \langle \varrho,\varrho\rangle_{2} = \sqrt{\mathrm{tr}(\varrho\varrho^{\dagger})}$ makes $\mathcal{B}_{2}(\mathbb{H})$ into a Hilbert space. To avoid notational overload, we use the same bracket $\langle\bullet,\bullet\rangle$ for all inner products.

In quantum mechanics, a state is described by an operator in the space of density operators $\mathcal{S}(\mathbb{H})$, defined as
{\small
\begin{align*}
    \mathcal{S}(\mathbb{H}) := \Big\{ \varrho \in \mathcal{B}_{1}(\mathbb{H}) \;\Big|\; \varrho \geq 0,\; \mathrm{tr}(\varrho)=1,\; \varrho=\varrho^{\dagger} \Big\}.
\end{align*}}

We set $\mathcal{B}_{0}(\mathbb{H})$ the space of bounded zero-trace operators:
{\small\begin{align*}
    \varrho \in \mathcal{B}_0(\mathbb{H}) \Longleftrightarrow \mathrm{tr}(\varrho) = 0.
\end{align*}
}

We let $\mathbf{1} $ denote the identity operator without explicitly specifying the underlying Hilbert space in use. 
For any   ${O}_{a},{O}_{b} \in \mathcal{B}(\mathbb{H})$, set $[{O}_{a},{O}_{b}]:= {O}_{a}{O}_{b} - {O}_{b}{O}_{a}$  and $\{{O}_{a},{O}_{b}\}:= {O}_{a}{O}_{b} + {O}_{b}{O}_{a}$.

For any operator ${O} \in \mathcal{B}(\mathbb{H})$ and for $l \in [N]$, denote by 
$\mathbf{O}_{l} := \mathbf{1}\otimes\dots\otimes{O}\otimes\dots\otimes\mathbf{1}$ the operator on $\mathcal{B}_{1}(\mathbb{H}^{\otimes N})$ that acts only on the $l$-th Hilbert space.

For any operator $B \in \mathcal{B}(\mathbb{H}\otimes\mathbb{H}) $ and for $l,l' \in \{1,\dots,N \}$ denote $\mathbf{B}_{ll'}$ the operator on $\mathcal{B}(\mathbb{H}^{\otimes N})$ that acts only on $\mathbb{H}_{l}$ and $\mathbb{H}_{l'}$.

The following matrices form the set of Pauli matrices:
{\small \begin{align*}
{\sigma}_x &:= \begin{pmatrix}0 & 1 \\ 1 & 0\end{pmatrix}, \quad
{\sigma}_y := \begin{pmatrix}0 & -\mathrm{i}\\ \mathrm{i} & 0\end{pmatrix}, \quad
{\sigma}_z := \begin{pmatrix}1 & 0 \\ 0 & -1\end{pmatrix}.
\end{align*}}

We note the ground and excited state: 
{\small \begin{align*}
    \rho_g := \begin{pmatrix}0 & 0 \\ 0 & 1\end{pmatrix}, \; 
    \rho_e := \begin{pmatrix}1 & 0 \\ 0 & 0\end{pmatrix}.\end{align*}}
We also set the following density matrices for the coupled system composed of two-level system
{\small \begin{align*}
    \rho_{ee} &= \rho_{e}\otimes\rho_{e},\; \rho_{eg} = \rho_{e}\otimes\rho_{g},\\ \rho_{gg} &= \rho_{g}\otimes\rho_{g},\; \rho_{ge} = \rho_{g}\otimes\rho_{e}.
\end{align*}}

\section{Quantum Filtering Formalism}
\label{sec:formalism}

Filtering theory allows for the continuous estimation of a system's state by deriving this dynamics from a unitary evolution, including the system plus the environment. In this framework, let $\mathbb{H}$ denote the system Hilbert space, and $\mathbb{D}$ the field algebra representing a bath. The quantum probability space is defined as $(\mathbb{A},\mathfrak{E})$, where $\mathbb{A} := \mathcal{B}(\mathbb{H})\otimes\mathbb{D}$, and $\mathfrak{E}$ is a state on $\mathbb{A}$. The time evolution of every observable $X = X^{\dagger} \in \mathcal{B}(\mathbb{H})$ coupled to a field is given by the unitary operator solution to a quantum stochastic differential equation \cite{parthasarathy12}:
{\small \begin{align*}
\mathrm{d}\mathsf{U}_t
  = \Big( \mathrm{d}\mathsf{B}_{t}^{\dagger}L
          - L^{\dagger}\mathrm{d}\mathsf{B}_{t}
          - \tfrac{1}{2} L^{\dagger}L\mathrm{d}t
          - \mathrm{i} \tilde{H} \mathrm{d}t \Big)\mathsf{U}_t,
\qquad \mathsf{U}_0 = \mathbf{1},
\end{align*}}
where $\tilde{H}=\tilde{H}^{\dagger}$ is the  Hamiltonian, $L\in\mathcal{B}(\mathbb{H})$
is the coupling operator, and $\mathsf{B}_t$ is a quantum annihilation process which satisfies quantum Ito's rule \cite{parthasarathy12}:
{\small \begin{equation*}
\begin{cases} 
\mathrm{d}\mathsf{B}_t\mathrm{d}\mathsf{B}_t = \mathrm{d}\mathsf{B}_t^{\dagger}\mathrm{d}\mathsf{B}_t = \mathrm{d}\mathsf{B}_t\mathrm{d}t = (\mathrm{d}t)^2 = 0, \\
\mathrm{d}\mathsf{B}_t\mathrm{d}\mathsf{B}^{\dagger}_t = \mathrm{d}t,
\end{cases}
\end{equation*}}
The evolution of the observable is {\small $
    X_t = \mathsf{U}_t^{\dagger}X\mathsf{U}_t$.} Under  homodyne detection \cite{gardiner04noise}. The measurement records correspond to the process {\small \begin{align*}
Y_t &= \mathsf{U}_t^{\dagger}(\mathsf{B}_t + \mathsf{B}_t^{\dagger})\mathsf{U}_t.
\end{align*}}
Consequently, the algebra generated by $(Y_t)_{t \geq 0}$, denoted by $\mathcal{Y}_t$ can be identified with a classical commutative algebra
{\small $L^\infty(\Omega,\mathcal{F}_t,\mathbb{P})$} for some filtered probability space
{\small $(\Omega,\mathcal{F},(\mathcal{F}_t)_{t \geq 0},\mathbb{P})$}, and the process $(Y_t)_{t\geq 0}$
is identified with an adapted real-valued  stochastic process.
 
\subsection{Filtering equation}
The filtering problem is now well posed in the sense that, for every observable $X \in \mathcal{B}(\mathbb{H})$,  $X_t = \mathsf{U}_t^{\dagger}X\mathsf{U}_t$, provides the optimal estimate of the observable given the
measurement record. We define
{$\pi_t(X) := \mathfrak{E}[X_t|\mathcal{Y}_t]
$.}

Using the equivalence between the Schrödinger and Heisenberg pictures, there exists a process $\rho_t \in \mathcal{S}(\mathbb{H})$ such that {\small $
     \pi_t(X) := \mathrm{tr}(\rho_tX)
$}.
Thanks to the cycle property of the trace, the quantum filtering equation \cite{boutenhandel07} can be derived: 
{\small
\begin{align}\label{eq:belavkin}
    \mathrm{d}\rho_t &= \mathcal{L}[\rho_t]\,\mathrm{d}t + \mathcal{R}[\rho_t]\mathrm{d}W_t,
\end{align}}
where {\small $W_t = Y_t - \int_{0}^{t}\mathrm{tr}((L+L^{\dagger})\rho_s)\mathrm{d}s $} is a standard one-dimensional Wiener process. 

The drift term $\mathcal{L}$, called the Lindbladian, is divided into two parts $\mathcal{L} = \mathcal{H} + \mathcal{D}$, where  $\mathcal{H}$ is the Hamiltonian part, and $\mathcal{D}$ is the dissipative part, corresponding to the coupling with the environment. The diffusive term $\mathcal{R}$ represents the measurement part, which encodes the observation records. These superoperators are given by 
{\small
\begin{align*}
    \mathcal{H}[\varrho] &:= -\mathrm{i}[\tilde{H},\varrho], \quad 
    \mathcal{D}[\varrho] := L\varrho L^{\dagger} - \frac{1}{2}\{L^{\dagger}L,\varrho\},\\
    \mathcal{R}[\varrho] &:= L\varrho + \varrho L^{\dagger} - \mathrm{tr}\big((L+L^{\dagger})\varrho\big)\varrho.
\end{align*}}
Taking the expectation $\bar{\rho}_{t}:=\mathbb{E}_{\mathbb{P}}[\rho_{t}]$, we recover the Lindblad equation \cite{GKLS17History}:
{\small
\begin{align}\label{eq:lindblad}
\frac{\mathrm{d}\bar{\rho}_t}{\mathrm{d}t} \;=\; \mathcal{L}[\bar{\rho}_t]
= -\mathrm{i}[\tilde{H},\bar{\rho}_{t}]
+ \Big(L\bar{\rho}_{t}L^{\dagger} - \tfrac{1}{2}\{L^{\dagger} L,\bar{\rho}_{t}\}\Big).
\end{align}}
The quantum filtering  equation is a quantum version of
a Kushner–Stratonovich equation, similarly the Lindblad equation \eqref{eq:lindblad} can be seen as a quantum analog of the Fokker-Planck-Kolmogorov equation, and the Lindbladian $\mathcal{L}$ the quantum analog of the adjoint of an infinitesimal generator of an SDE. 
\subsection{Optimal control}
In the context of quantum feedback control, the experimenter has access to the measurement record $(Y_t)_{t \ge 0}$ and can use this information to modulate the system dynamics in real time. The control variable, denoted by $\alpha_t$, typically acts on the system through an external field or a tunable Hamiltonian term. 

In this case the total Hamiltonian is denoted by
$\mathcal{H}[\varrho, \beta]$ and defined as follows
{\small
\begin{align*}
    \mathcal{H}[\varrho, \beta] := -\mathrm{i}[\tilde{H} + \beta\widehat{H},\varrho],
\end{align*}}
where $\widehat{H} = \widehat{H}^{\dagger} \in \mathcal{B}(\mathbb{H})$ corresponds to the control operator associated with the actuator. The full controlled Lindbladian becomes ${\small \mathcal{L}[\varrho, \beta] = \mathcal{H}[\varrho, \beta] + \mathcal{D}[\varrho]}$.

The corresponding controlled quantum filtering equation governing the conditional state $\rho_t$ under the control $\alpha_t$ reads:
{\small
\begin{align}\label{eq:contbelavkin}
    \mathrm{d}\rho_t 
    = \mathcal{L}[\rho_t, \alpha_t]\mathrm{d}t
    + \mathcal{R}[\rho_t]\mathrm{d}W_t,\quad
    \rho_0 = \varrho \in \mathcal{S}(\mathbb{H}).
\end{align}}
The control process $\{\alpha_t\}_{t \ge 0}$ is assumed to be adapted to the measurement
filtration generated by the record process, ensuring causality and the implementability of the feedback strategy.

In the classical setting, the optimal feedback control problem is solved by coupling the HJB equation with the Kushner-Stratonovich equation. This approach is carried out in an analogous manner in the quantum regime \cite{belavkin09cybernetics,gough05hamilton,belavkin79optimal}, where dynamic programming was used. In the present work, we rigorously establish the dynamic programming principle in infinite dimensional setting. Let us nonetheless note that in the case of a finite-dimensional Hilbert space $\mathbb{H}=\mathbb{C}^{d}$, the space of operators $\mathcal{B}(\mathbb{H})=\mathcal{M}_d(\mathbb{C})$ is isomorphic to $\mathbb{R}^{2d^2}$. Consequently, these spaces can be identified with standard Euclidean spaces, allowing one to invoke classical results of stochastic control. However, in the infinite-dimensional setting, the situation is more delicate. Although the Banach space of trace-class operators $\mathcal{B}_1(\mathbb{H})$ appears to be the natural framework when we work with density operators, stochastic calculus is not well-defined therein\footnote{Extensions of Itô stochastic calculus have been developed for some classes of Banach spaces through the theory of UMD spaces \cite{UMD15}; however, the space of trace-class operators does not belong to those classes.}. Therefore, it is more convenient to work within the space of Hilbert--Schmidt operators $\mathcal{B}_2(\mathbb{H})$, which is a Hilbert space and where Itô theory is established.

We now introduce differentiability of functionals on the Hilbert--Schmidt space  $\mathcal{B}_{2}(\mathbb{H})$.
\begin{definition}[Differentiability]
Let $\mathrm{G}: \mathcal{B}_{2}(\mathbb{H}) \to \mathbb{R}$ be a functional. We say that $\mathrm{G}$ admits a Fréchet derivative if there exists an $\mathcal{B}(\mathbb{H})$-valued function $\nabla \mathrm{G}[\bullet]$ on $\mathcal{B}_{2}(\mathbb{H})$ such that 
{\small \begin{align*}
    \lim_{\varepsilon \to 0}\frac{1}{\varepsilon}\Big(\mathrm{G}[\bullet + \varepsilon \tau] - \mathrm{G}[\bullet] \Big) = \big\langle \tau, \nabla \mathrm{G}[\bullet] \big\rangle,  
\end{align*}}
for each {\small $\tau \in \mathcal{B}_{0}(\mathbb{H})\cap\mathcal{B}_{2}(\mathbb{H})$}. In the same manner, a Hessian {\small $\nabla^{\otimes 2} \equiv \nabla \otimes \nabla$} is defined as a mapping from functionals on {\small $\mathcal{B}_{2}(\mathbb{H})$} to {\small $\mathcal{B}(\mathbb{H}^2)$}, via 
{\small \begin{align*}
    \lim_{\varepsilon,\varepsilon' \to 0} \frac{1}{\varepsilon \varepsilon'} &\Big( \mathrm{G}[\bullet + \varepsilon{\tau} + \varepsilon'\tau' ] - \mathrm{G}[\bullet + \varepsilon{\tau}] - \mathrm{G}[\bullet + \varepsilon'{\tau}'] + \mathrm{G}[\bullet]\Big)  \\
    &=\big\langle \tau \otimes \tau' , \nabla^{\otimes 2}\mathrm{G}[\bullet] \big\rangle 
\end{align*}}
for each {\small $\tau, \tau' \in \mathcal{B}_{0}(\mathbb{H})\cap\mathcal{B}_{2}(\mathbb{H}).$}

We say that the functional $\mathrm{G}$ is continuously differentiable (resp. twice continuously differentiable) whenever $\nabla\mathrm{G}[\bullet]$  (resp. $\nabla^{\otimes 2}\mathrm{G}[\bullet]$) exists and is continuous in the Hilbert-Schmidt norm topology.
\end{definition}
Viewed as a function of the variable $\varrho$, the mappings 
$\mathcal{L}$ and $\mathcal{R}$ are Lipschitz continuous in the trace-norm topology on $\mathcal{B}_1(\mathbb{H})$. 
It follows that $\mathcal{L}$ is automatically continuous when regarded as mappings on the Hilbert--Schmidt space $\mathcal{B}_2(\mathbb{H})$. However, this is  not the case for $\mathcal{R}$ in general. In the following we restrict to the coupling operator $L$ of Hilbert-Schmidt form, the mapping $\mathcal{R}$ is Lipschitz continuous.
\begin{lemma}\label{cond-lipch}
If  {\small $L \in \mathcal{B}_{2}(\mathbb{H})$}, then  the non-linear mapping $\mathcal{R}$ is Lipschitz continuous with respect to the Hilbert--Schmidt norm. 
\end{lemma}
\begin{proof}
Let $\rho, \rho' \in \mathcal{S}(\mathbb{H})$ and denote $\delta = \rho - \rho'$, set $K = L+L^{\dagger}$. We have
{\small
\begin{align*}
    \|\mathcal{R}[\delta]\|_2 \leq \|L\delta + \delta L^\dagger\|_2 + \|\mathrm{tr}(K\rho)\rho - \mathrm{tr}(K\rho')\rho'\|_2.
\end{align*}}
For the linear part, using the property $\|AB\|_2 \leq \|A\|\|B\|_2$ and the fact that $\|L\| \leq \|L\|_2 < \infty$, we have
{\small
\begin{align*}
    \|L\delta + \delta L^\dagger\|_2 \leq 2\|L\|_{2} \|\rho - \rho'\|_2.
\end{align*}}
For the nonlinear part, we use the standard decomposition
{\small
\begin{align*}
    \mathrm{tr}(K\rho)\rho - \mathrm{tr}(K\rho')\rho' &= \mathrm{tr}(K\rho)\delta + \mathrm{tr}(K\delta)\rho' \\
    &= \mathrm{tr}(K\rho)\delta + \mathrm{tr}(K\delta)\rho'.
\end{align*}}
Taking the Hilbert--Schmidt norm and applying the Cauchy-Schwarz inequality, we obtain
{\small
\begin{align*}
    \|\mathrm{tr}(K\rho)\rho - \mathrm{tr}(K\rho')\rho'\|_2
    &\leq |\mathrm{tr}(K\rho)| \|\delta\|_2 + |\mathrm{tr}(K\delta)| \|\rho'\|_2 \\
    &\leq (\|K\|_2 \|\rho\|_2) \|\delta\|_2 + (\|K\|_2 \|\delta\|_2) \|\rho'\|_2 \\
    &= \|K\|_2 (\|\rho\|_2 + \|\rho'\|_2) \|\rho - \rho'\|_2.
\end{align*}}
Combining these estimates, we find
{\small
\begin{align*}
    \|\mathcal{R}[\delta]\|_2 &\leq \Big( 2\|L\|_2 + \|L+L^\dagger\|_2 (\|\rho\|_2 + \|\rho'\|_2) \Big) \|\delta\|_2\\
    &\leq 6\|L\|_{2}\|\delta\|_{2}.
\end{align*}}
\end{proof}

For a functional $\mathrm{G}$ that is twice continuously differentiable on $\mathcal{B}_2(\mathbb{H})$, we have
{\small
\begin{align*}
\lim_{\varepsilon \to 0} \frac{1}{\varepsilon}\Big\{ \mathbb{E}\Big[\mathrm{G}[\rho_{t+\varepsilon}] - \mathrm{G}[\rho_{t}] \Big]\Big\} &= \mathrm{D}(\rho_t,\alpha_t)\big[\mathrm{G}[\rho_t]\big], 
\end{align*}
}
where $\mathrm{D}(\varrho,\beta)$ is the elliptic operator defined by
{\small 
\begin{align*}
    \mathrm{D}(\varrho,\beta)[\bullet] := \langle \mathcal{L}(\varrho,\beta), \nabla[\bullet]\rangle + \frac{1}{2}\langle \mathcal{R}(\varrho)\otimes\mathcal{R}(\varrho),\nabla^{\otimes 2}[\bullet]\rangle.
\end{align*}}
To specify the control objective, we associate to each admissible control 
$\alpha=\{\alpha_s\}_{s\ge t}$ the cost functional
{\small 
\begin{align*}
    \mathcal{S}(t,\varrho, \alpha) := \mathbb{E}\big[\int_{t}^{T}\mathrm{C}(\rho_s, \alpha_s)\mathrm{d}s + \mathrm{F}(\rho_T)\big].
\end{align*}}
In applications, it is common to take the terminal cost 
$\mathrm{F}: \mathcal{B}_2(\mathbb{H}) \to \mathbb{R}$ linear in
$\varrho$ and the running cost 
$\mathrm{C}: \mathcal{B}_2(\mathbb{H})\times \mathcal{U} \to \mathbb{R}$
quadratic in $\varrho$ and in the control  $\beta$. In
contrast to \cite{gough05hamilton,belavkin79optimal}, where only linear costs
were considered, we allow for quadratic terms. For example,
{\small
\begin{align*}
    \mathrm{C}(\varrho,\beta) &= \langle \varrho,\mathrm{C}_1\varrho\rangle + \beta^{2}\langle \varrho, C_2\rangle, \quad C_1,C_2 \in \mathcal{B}_{2}(\mathbb{H})\\
    \mathrm{F}(\varrho) &= \langle F, \varrho \rangle, \quad F \in \mathcal{B}_2(\mathbb{H}).
 \end{align*}}
The associated value function is defined by
{\small
\begin{align*}
\mathsf{V}(t,\varrho) := \inf_{ \alpha \in \mathcal{U}}\mathcal{S}(t, \varrho, \alpha)
\end{align*}}

So now we can state the Dynamic Programming Principle.
\begin{theorem} Let $t \in [0,T]$ and $\varrho\in \mathcal{S}(\mathbb{H})$.
For every stopping time $\tau$ with $t\le \tau\le T$, the value function satisfies
{\small\begin{align*}
    \mathsf{V}(t,\varrho) &= \inf_{ \alpha \in \mathcal{U}} \mathbb{E}\bigg[\int_{t}^{\tau}\mathrm{C}(\rho_s^{t},\alpha_s)\mathrm{d}s + \mathsf{V}(\tau,\rho_{\tau}^{t})\bigg],
\end{align*}}
 where $\rho^{t}$ denotes the solution of~\eqref{eq:contbelavkin} starting at time $t$ with initial condition $\rho_t=\varrho$ and control $\alpha$.
\end{theorem}
\begin{proof}
For any admissible control $\alpha\in\mathcal{U}$ and any $\tau\in[t,T]$,
{\small
\begin{align*}
    \mathcal{S}(t,\varrho,\alpha)
    &= \mathbb{E}\bigg[\int_{t}^{\tau}
       \mathrm{C}\big(\rho_s^{t},\alpha_s\big)\mathrm{d}s\bigg]
     \! + \! \mathbb{E}\bigg[\int_{\tau}^{T}
       \mathrm{C}\big(\rho_s^{t},\alpha_s\big)\mathrm{d}s
       + \mathrm{F}\big(\rho_T^{t}\big) \bigg]\\
    &\geq \mathbb{E}\bigg[\int_{t}^{\tau}
       \mathrm{C}\big(\rho_s^{t},\alpha_s\big)\mathrm{d}s
       + \mathsf{V}\big(\tau,\rho_{\tau}^{t}\big)\bigg],
\end{align*}}
since the second expectation is bounded below by
$\mathsf{V}(\tau,\rho_{\tau}^{t})$ by definition of
$\mathsf{V}$. Taking the infimum over $\alpha\in\mathcal{U}$ yields
{\small
\begin{align*}
    \mathsf{V}(t,\varrho)
    = \inf_{\alpha\in\mathcal{U}} \mathcal{S}(t,\varrho,\alpha) \geq \inf_{\alpha\in\mathcal{U}}
    \mathbb{E}\bigg[\int_{t}^{\tau}
    \mathrm{C}\big(\rho_s^{t},\alpha_s\big)\mathrm{d}s
    + \mathsf{V}\big(\tau,\rho_{\tau}^{t}\big)\bigg].
\end{align*}}

For the reverse inequality, fix $\varepsilon>0$ and $\alpha\in\mathcal{U}$.
By definition of $\mathsf{V}$, there exists a control
$\alpha^{\varepsilon} \in \mathcal{U}$ such that
{\small
\begin{align*}
\mathcal{S}\big(\tau,\rho_{\tau}^{t},\alpha^{\varepsilon}\big)
    \leq \mathsf{V}\big(\tau,\rho_{\tau}^{t}\big) + \varepsilon.
\end{align*}}
Define the concatenated control
{\small
\begin{align*}
    \tilde{\alpha}_s^{\varepsilon}
    := \alpha_s\mathbb{I}_{[t,\tau)}(s)
     + \alpha_s^{\varepsilon}\mathbb{I}_{[\tau,T]}(s),
     \quad s \in [0,T].
\end{align*}}
Then, using the flow property of the controlled dynamics,
{\small
\begin{align*}
    \mathcal{S}(t,\varrho,\tilde{\alpha}^{\varepsilon})
    &= \mathbb{E}\bigg[\int_{t}^{\tau}
       \mathrm{C}\big(\rho_s^{t},\alpha_s\big)\mathrm{d}s\bigg]
     + \mathbb{E}\bigg[\mathcal{S}\big(\tau,\rho_{\tau}^{t},
       \alpha^{\varepsilon}\big)\bigg]\\
    &\leq \mathbb{E}\bigg[\int_{t}^{\tau}
       \mathrm{C}\big(\rho_s^{t},\alpha_s\big)\,\mathrm{d}s
       + \mathsf{V}\big(\tau,\rho_{\tau}^{t}\big)\bigg]
       + \varepsilon,
\end{align*}}
and therefore,
{\small
$
\mathsf{V}(t,\varrho) \leq  \mathbb{E}\bigg[\int_{t}^{\tau}
       \mathrm{C}\big(\rho_s^{t},\alpha_s\big)\,\mathrm{d}s
       + \mathsf{V}\big(\tau,\rho_{\tau}^{t}\big)\bigg]
       + \varepsilon.
$}

Taking the infimum over all $\alpha\in\mathcal{U}$ and letting
$\varepsilon\to 0$ yields the desired equality.
\end{proof}
\paragraph*{Hamilton-Jacobi-Bellman equation}
Thanks to the  DPP, we can now derive the HJB equation formally. For $\varepsilon$ small, thanks to Lemma \ref{cond-lipch} and applying Itô's formula to the process $\mathsf{V}(s, \rho_s)$ between $t$ and $t+\varepsilon$ yields:
{\small
\begin{align*}
    \mathsf{V}(t,\varrho) &=\inf_{\alpha \in \mathcal{U}}\mathbb{E}\Big[\int_{t}^{t+\varepsilon}\mathrm{C}(\rho_s,\alpha_s)\mathrm{d}s + \mathsf{V}(t+\varepsilon,\rho_{t+\varepsilon})\Big]\\
    &= \inf_{\alpha \in \mathcal{U}}\mathbb{E}\bigg[\int_{t}^{t+\varepsilon} \mathrm{C}(\rho_s, \alpha_s)\mathrm{d}s + \mathsf{V}(t,\varrho) \\
    &\quad + \int_{t}^{t+\varepsilon}\Big(\frac{\partial \mathsf{V}}{\partial t}(s, \rho_s) + \mathsf{D}(\rho_s,\alpha_s)[\mathsf{V}(s,\cdot)]\Big)\mathrm{d}s\bigg].
\end{align*}}
Subtracting $\mathsf{V}(t,\varrho)$ from both sides, dividing by $\varepsilon$ and letting $\varepsilon \to 0^{}$, gives the HJB equation
{\small
\begin{align*}
   0 &= \inf_{\alpha \in \mathcal{U}}\bigg\{ \mathrm{C}(\varrho,\alpha) + \frac{\partial \mathsf{V}}{\partial t}(t,\varrho) + \mathrm{D}(\varrho,\alpha)[\mathsf{V}(t,\cdot)] \bigg\}.
\end{align*}}
Rearranging the terms leads to the standard form:
{\small
\begin{align*}
   -\frac{\partial\mathsf{V}}{\partial t}(t,\varrho)   &= \inf_{\alpha \in \mathcal{U}}\big\{ \mathrm{C}(\varrho,\alpha) + \mathrm{D}(\varrho,\alpha)[\mathsf{V}(t,\varrho)] \big\}.
\end{align*}}
The equation is solved backward in time subject to the terminal condition: {\small
$    \mathsf{V}(T,\varrho) = \mathrm{F}(\varrho).
$}

The formal derivation of the HJB equation above suggests that a smooth solution, should it exist, coincides with the value function. Consequently, the next natural step is to establish a verification theorem, which ensures that a candidate solution to the HJB equation indeed matches the value function. However, in many practical scenarios, the HJB equation may not admit a classical solution (i.e., one with sufficient regularity). To address this, one can adopt a weaker notion of solution known as a viscosity solution, which provides a rigorous framework for handling non-smooth solutions to partial differential equations. The full development of these concepts is left for future work.

\section{N-Body and Mean-field dynamics}
\label{sec:nbody}
Having reviewed the theory of quantum filtering and feedback control, we now introduce the
$N$-body formalism and its mean-field limits. This provides the foundation for quantum
differential games and mean-field quantum control. In this setting, the filtering framework
extends naturally to composite systems whose Hilbert space consists of $N$ identical
quantum subsystems:
{\small \begin{align*}
    \mathbb{H}^{N} &:= \mathbb{H}\otimes\dots\otimes\mathbb{H},
\end{align*}}
where $\mathbb{H} = L^2_{\mathbb{C}}(\mathfrak{X})$, and the environment (bath) is modeled by $N$-quantum fields. The conditional state for $N$-Body system follows quantum filtering equation:
{\small
\begin{align}\label{Nbelavkin}
\mathrm{d}\boldsymbol{\rho}_t^{N} &= -\mathrm{i}{[\mathbf{H}^{N},\boldsymbol{\rho}_{t}^{N}]}\mathrm{d}t
+ \sum_{l = 1}^{N} \Big({\bf L}_l\boldsymbol{\rho}_{t}^{N}{\bf L}^{\dagger}_{l} - \frac{1}{2}\big\{{{\bf L}}^{\dagger}_{l}{\bf L}_{l},\boldsymbol{\rho}_{t}^{N}\big\}\Big)\mathrm{d}t\nonumber\\
&\; +\sum_{l = 1}^{N}\Big({\bf L}_l\boldsymbol{\rho}_{t}^{N} + \boldsymbol{\rho}_{t}^{N}{\bf L}^{\dagger}_l -\mathrm{tr}\big(({\bf L}_l + {\bf L}_l^{\dagger})\boldsymbol{\rho}_{t}^{N}\big)\boldsymbol{\rho}_{t}^{N}\Big)\mathrm{d}W_t^{l}.
\end{align}}

The output process for each sub-systems is then:
{\small \begin{align*}
    \mathrm{d}Y_t^{l} = \mathrm{d}W_t^{l} + \mathrm{tr}\big((\mathbf{L}_{l}+\mathbf{L}_{l}^{\dagger})\boldsymbol{\rho}_{t}^{N}\big)\mathrm{d}t, \;\; l \in [N].
\end{align*}}

In the case of  pairwise interactions between  sub-systems, the Hamiltonian part has the following structure: {\small $
\mathbf{H}^{N} = \sum_{l=1}^{N}\tilde{\mathbf{H}}_{l}+ 
\frac{1}{N} \sum_{l > l'}^{N} \mathbf{A}_{ll'}.$}

The pairwise interaction Hamiltonian is an operator 
on $L^2_{\mathbb{C}}(\mathfrak{X}^2,\mu^{\otimes 2})$, 
of Hilbert--Schmidt type with kernel $a$. This kernel satisfies:
{\small$$
a(x,y;x',y') = a(y,x;y',x'), 
\; 
a(x,y;x',y') = \overline{a(x',y';x,y)}.
$$}

The corresponding operator $A$ acts on 
$L^2_{\mathbb{C}}(\mathfrak{X}^2,\mu^{\otimes 2})$ 
as follows: for all $f \in L^2_{\mathbb{C}}(\mathfrak{X}^2,\mu^{\otimes 2}), $ we have
{\small\begin{align*}
\bigl(Af\bigr)(x,y) 
:= 
\int_{\mathfrak{X}^2} 
a(x,y;x',y')f(x',y') 
\mu(\mathrm{d}x')\mu(\mathrm{d}y').
\end{align*}}

Everything is now in place to construct an \emph{$N$-quantum differential game}.  
We attach to each particle $l \in [N]$ a distinct \emph{agent} responsible for feedback control of that particle.  
Agent $l$ observes its \emph{local state} 
obtained by tracing out all other particles from the global state $\boldsymbol{\rho}_t^N$. Based on this information, agent $l$ applies a feedback control $\alpha_l \in \mathcal{U}$ that depends on its local state  $\alpha_{t,l} := \alpha_{l}(\rho^{l}_t)$.

Hence, the total Hamiltonian becomes time-dependent:
{\small \begin{align*}
\mathbf{H}^{N}_t = \sum_{l=1}^{N}\tilde{\mathbf{H}}_{l}+ 
\frac{1}{N} \sum_{l > l'}^{N} \mathbf{A}_{ll'} + \sum_{l=1}^{N}\alpha_{t,l}\widehat{\mathbf{H}}_{l}.
\end{align*}}
Each player $l \in [N]$ seeks to minimize its individual cost functional 
{\small \begin{align*}\mathcal{J}_{l}\big(\alpha_l,\boldsymbol{\alpha}_{-l}\big) &= \mathbb{E}\Big[\int_{0}^{T} \mathrm{C}_l(\boldsymbol{\rho}^{N}_{s},\alpha_{s,l})\mathrm{d}s + \mathrm{F}_l(\boldsymbol{\rho}^{N}_{T})\Big],
\end{align*}}%
where $\boldsymbol{\alpha}_{-l}$ denotes the controls of all other players.
An optimal set of control policy leads to an $N$-player quantum differential games. As in the classical case, and even more pronounced in the quantum setting due to the multiplication of dimensions, complexity grows rapidly and the system quickly becomes unsolvable as $N$ increases. However, due to the complete symmetry of interactions, if we let $N \to \infty$, the interaction weakens in the sense that the influence of each particle, which is of order $1/N$, goes to $0$. As a result, we expect the interactions to vanish and the sub-systems become asymptotically uncorrelated.  This is the idea behind the mean-field limit; the mathematical justification of the mean-field limit is done through the notion of propagation of chaos. Basically, the statement says that if the system is initially in a chaotic (i.e., tensor-product) state, it remains so at later times. More precisely,
{\small \begin{align*}
    \boldsymbol{\rho}_{0}^{N} = \rho_{0}^{\otimes N} \xrightarrow[]{\text{Propagation of chaos}} \boldsymbol{\rho}_{t}^{N}  \approx \bigotimes_{l=1}^{N} \gamma_{t,l}. 
\end{align*}}

In this limit, a generic subsystem evolves according to the mean-field quantum filtering equation
{\small
\begin{align}\label{mfbelavkin}
    \mathrm{d}\gamma_{t} &= -\mathrm{i}[\tilde{H} + A^{\mathbb{E}[\gamma_{t}]},\gamma_{t}]\mathrm{d}t + \big(L\gamma_{t}L^{\dagger} - \frac{1}{2}\{L^{\dagger}L,\gamma_{t}\}\big)\mathrm{d}t\nonumber\\ 
    &\quad + \Big(L\gamma_{t} + \gamma_{t}L^{\dagger} - \mathrm{tr}\big((L+L^{\dagger})\gamma_{t}\big)\gamma_{t}\Big)\mathrm{d}W_t,
\end{align}}
with the initial state $\gamma_{0} = \rho_{0} $. The corresponding mean-field operator $A^{\bullet}$ is the linear map
{\small \begin{align}\label{mfoperator}
A^{\bullet} &: L^{2}_{\mathbb{C}}(\mathfrak{X}\times\mathfrak{X},\mu^{\otimes 2}) \to L^{2}_{\mathbb{C}}(\mathfrak{X}\times\mathfrak{X},\mu^{\otimes 2}),\nonumber\\
A^{\rho}(x,x') &= \int_{\mathfrak{X}^{2}}a(x,y;x',y')\overline{\rho(y,y')}\mu(\mathrm{d}y)\mu(\mathrm{d}y'),
\end{align}}
and the output process for the mean-field particle is 
{\small \begin{align}\label{recordprocess}
    \mathrm{d}Y_t^{} = \mathrm{d}W_t^{} +\mathrm{tr}\big(({L}_{}+{L}_{}^{\dagger}){\gamma}_{t}\big)\mathrm{d}t.
\end{align}}

\begin{remark}
The derivation of \eqref{mfbelavkin}, from the $N$-body dynamics was first carried out rigorously in \cite{kolokoltsov21law,kolokoltsov22qmfg}; see also \cite{kolokoltsov25quantumstoeq,chalal23mean,kolokoltsov25mathematical} for questions related to well-posedness.  
\end{remark}

Coupling this limit system with an optimal control policy leads to a quantum mean-field feedback control problem.

Now that the formalism has been established, the next section presents a concrete example of quantum state preparation that can be viewed as a differential game, and for which the mean-field approximation provides an effective method of dimensionality reduction.

\section{Game of Stabilization}
\label{sec:game}

Stabilization is one of the central tasks in control engineering. Given a desired equilibrium point\footnote{By an equilibrium point we mean a state of the system that is invariant under the dynamics: if the system is initialized at this state, it remains there for all future times.}, the goal is to design a feedback control that makes this equilibrium asymptotically stable for the controlled system; that is, solutions starting close to the equilibrium remain close, and all solutions eventually converge to the equilibrium.

It is well known that optimal control methods can be used for the design of asymptotically stabilizing controls by choosing the objective in such a way that it penalizes states away from the desired equilibrium. Suppose the system starts from an initial state $\rho_{0} \in \mathcal{S}(\mathbb{H})$ and we wish it to reach, at some terminal time $T$, a target state $\rho^{\star} \in \mathcal{S}(\mathbb{H})$. 
 To model this, we assign to each state $\rho \in \mathcal{S}(\mathbb{H})$ a terminal cost $\mathrm{F}(\rho) \in \mathbb{R}$, which is small when $\rho$ is close to the target and large otherwise, so that the controller is incentivized to minimize $\mathrm{F}(\rho_{T})$. If one only considers controllability, any control $\alpha \in \mathcal{U}$ that steers $\rho_{0}$ toward $\rho^{\star}$ would be acceptable. However, in many situations it is more appropriate to penalize not only the terminal deviation but also the “transportation cost’’ incurred along the trajectory. This leads to the following cost functional:
{\small \begin{align*}\mathcal{J}(\alpha) = \mathbb{E}\Big[\int_{0}^{T}\mathrm{C}(\rho_s,\alpha_s)\mathrm{d}s + \mathrm{F}(\rho_{T})  \Big]. \end{align*}}

The question of quantum state preparation and quantum state stabilization is an active research area in quantum control that has been investigated in several papers   \cite{stockton04deter,liang19exponential,amini24exponential}. It consists in designing control laws that steer the quantum state toward a desired target.  In the following,  we consider the case of multiple controllers (agents), each with their own objective. We  illustrate how this setting can be interpreted as a quantum differential game, and then show exponential stabilization in the  mean-field setting.

In the following we consider $N$ qubits with pairwise Ising interactions. The configuration space is {\small $\mathfrak{X} = \{1,2\}$}, endowed with the counting measure, such that $\mathbb{H} := L^{2}_{\mathbb{C}}(\mathfrak{X}) \cong \mathbb{C}^2$. Each agent $l \in [N]$ continuously monitors its qubit along the $z$-axis and applies a local Hamiltonian feedback along the $x$-axis:
{\small\begin{align*}
A &=  {\sigma}_{z}\otimes{\sigma}_{z}, 
\quad L = {\sigma}_{z},\quad  \widehat{H} = \sigma_{x},\quad \tilde{H} = 0,  \quad l \in [N],\\
\mathrm{d}\boldsymbol{\rho}_t^{N}
&= -\mathrm{i}[\tfrac{1}{N}\sum_{1 \leq l < l' \leq N}\mathbf{A}_{ll'} + \sum_{l=1}^{N}\alpha_{t,l}\widehat{\mathbf{H}}_{l},\boldsymbol{\rho}^{N}_t]\mathrm{d}t\\
&\quad+ \sum_{l=1}^{N}\big(\boldsymbol{\sigma}_{z,l}\boldsymbol{\rho}_t\boldsymbol{\sigma}_{z,l}-\boldsymbol{\rho}_t\big)\mathrm{d}t\\
&\quad+ \sum_{l=1}^{N}\Big(\boldsymbol{\sigma}_{z,l}\boldsymbol{\rho}_t+\boldsymbol{\rho}_t\boldsymbol{\sigma}_{z,l}
 - 2\mathrm{tr}(\boldsymbol{\sigma}_{z,l}\boldsymbol{\rho}_t)\boldsymbol{\rho}_t\Big)\mathrm{d}W_t^{l},
\end{align*}}
with $\boldsymbol{\rho}_0\in\mathcal{S}(\mathbb{C}^{2^N})$.

\subsection{Mean-field setting}
If the agents share a common objective, then—due to the complete symmetry of the model—a mean-field formulation is natural.  
In this case the mean-field operator is $A^{\bullet} = \mathrm{tr}(\bullet\sigma_z)\sigma_z$. Indeed, for the Ising coupling operator $A = \sigma_z\otimes\sigma_z$, its associated kernel can be identified directly as
{\small
\begin{align*}a &= \small
{\begin{pmatrix}
        a(1,1;1,1) & a(1,1;2,1) & a(1,1;1,2) & a(1,1;2,2)\\
        a(2,1;1,1) & a(2,1;2,1) & a(2,1;1,2) & a(2,1;2,2)\\
        a(1,2;1,1) & a(1,2;2,1) & a(1,2;1,2) & a(1,2;2,2)\\
        a(2,2;1,1) & a(2,2;2,1) & a(2,2;1,2) & a(2,2;2,2)\\
    \end{pmatrix}}\\
    &= {\small \begin{pmatrix}
        1 & 0 & 0 & 0\\
        0 & -1 & 0 & 0\\
        0 & 0 & -1 & 0\\
        0 & 0 & 0 & 1\\
    \end{pmatrix}}
\end{align*}}
Then using \eqref{mfoperator}, {\small $
 A^{\rho}(x,x') = \sum_{(y,y')\in \mathfrak{X}^2}a(x,y;x',y')\overline{\rho(y,y')}
$}
such that, 
{\small
\begin{align*}
    A^{\rho}(1,1) &= \sum_{(y,y')\in \mathfrak{X}^2}a(1,y;1,y')\overline{\rho(y,y')},\\
    A^{\rho}(1,2) &= \sum_{(y,y')\in \mathfrak{X}^2}a(1,y;2,y')\overline{\rho(y,y')},\\
    A^{\rho}(2,1) &= \sum_{(y,y')\in \mathfrak{X}^2}a(2,y;1,y')\overline{\rho(y,y')},\\
    A^{\rho}(2,2) &= \sum_{(y,y')\in \mathfrak{X}^2}a(2,y;2,y')\overline{\rho(y,y')}.
\end{align*}}
In particular,
{\small \begin{align*}
A^{\rho}(1,1) &= \underbrace{a(1,1;1,1)}_{1}\overline{\rho(1,1)} + \underbrace{a(1,2;1,1)}_{0}\overline{\rho(2,1)} \\ & \quad + \underbrace{a(1,1;1,2)}_{0}\overline{\rho(1,2)} + \underbrace{a(1,2;1,2)}_{-1}\overline{\rho(2,2)}\\
&= \rho(1,1) - \rho(2,2).
\end{align*}}
Writing in Pauli parametrization, we get {\small $
A^{\rho}(1,1) = z $}, {\small $
    A^{\rho}(2,1) = 0, \; A^{\rho}(1,2) = 0, \; A^{\rho}(2,2) = -z.$}
Thus,
{\small \begin{align*}
A^{\rho} &= \begin{pmatrix}
    z & 0 \\ 0 & -z
\end{pmatrix} = z\sigma_{z} = \mathrm{tr}(\rho \sigma_{z})\sigma_z.
\end{align*}}

Then the dynamics is written as follow:
{\small\begin{align}\label{eq:mf-ising}
\mathrm{d}\gamma_t 
&= -\mathrm{i}\Big[\mathrm{tr}(\mathbb{E}[\gamma_{t}]\sigma_{z}) {\sigma}_z + \alpha_t\sigma_x,\gamma_t\Big]\mathrm{d}t
 + \big({\sigma}_z\gamma_t{\sigma}_z-\gamma_t\big)\mathrm{d}t\nonumber\\
&\quad + \big({\sigma}_z\gamma_t+\gamma_t{\sigma}_z - 2\mathrm{tr}({\sigma}_z\gamma_t)\gamma_t\big)\mathrm{d}W_t.
\end{align}}

In Pauli basis, the Bloch sphere components satisfy
{\small\begin{align}\label{eq:isingeq}
\mathrm{d}x_t &= \big(-2\mathbb{E}[z_t]y_t - 2x_t\big)\mathrm{d}t - 2z_tx_t\mathrm{d}W_t\nonumber\\
\mathrm{d}y_t &= \big(2\mathbb{E}[z_t]x_t - 2y_t - 2\alpha_t z_t\big)\mathrm{d}t - 2z_ty_t\mathrm{d}W_t\\
\mathrm{d}z_t &= \big(2\alpha_t y_t\big)\mathrm{d}t + 2(1-z_t^{2})\mathrm{d}W_t.\nonumber
\end{align}}

\begin{remark}
The presence of the mean-field term makes the treatment of the equation more delicate than the traditional case of a single controlled qubit in the literature. In the latter case, the variable $x_t$ has no dynamics as long as $x_0 = 0$; such that we can restrict  the analysis to the pair $(y_t,z_t)$ and stabilize the target point $(y,z) = (0,-1)$. The geometry of state space will guarantee the stabilization of the system. In the mean-field case, the fundamental feature of quantum state reduction remains valid, which is not surprising since the mean-field operator, although state-dependent, still commutes with the coupling operator $\sigma_z$. However, proving that an adequate control input stabilizes the desired eigenstate of the coupling operator requires care, because standard stochastic control lemmas are not stated for McKean-Vlasov SDEs. To overcome this issue, we first replace the expectation dependence by a time-dependent function valued in the set of density operators, then show stabilization for all such parametrization and finally apply a  fixed point argument to conclude the stabilization property for the mean-field equation.  
\end{remark}

\begin{theorem}\label{thm:stab}
Consider the system \eqref{eq:isingeq}. Then the following asymptotic properties hold: 
\begin{itemize}

\item \textbf{Quantum State Reduction.} Under zero feedback control $\alpha \equiv 0$,  the process $(\gamma_t)_{t \geq 0}$ converges almost surely,
as $t\to\infty$, to one of the eigenstates of $\sigma_z$, namely
$\rho_e$ or $\rho_g$.

\item \textbf{Stabilization. }  Let $\varrho^{\star} \in \{\rho_g,\rho_e\}$ with the feedback control
{\small \begin{align*}
    \alpha^{\star}_t &= \kappa_{1}\big( 1 - \mathrm{tr}(\gamma_{t}\varrho^{\star})\big) - \mathrm{i}\kappa_{2}\mathrm{tr}([\sigma_y,\gamma_t]\varrho^{\star}),
\end{align*}}
where $\kappa_{1} > 0, \kappa_{2} \geq 0.$ Then the state $\gamma_t$ converges exponentially toward $\varrho^{\star}$.
\end{itemize}
\end{theorem}

{

\begin{proof}

We first prove quantum state reduction.  
When the control is turned off, it is sufficient to focus on the dynamics of $z_t$, which satisfies
$\mathrm{d}z_t = 2(1-z_t^2)\mathrm{d}W_t$.
Since $z_t$ is bounded and is a local martingale, Doob's martingale convergence  theorem implies that it converges a.s. to one of its equilibirium points, namely $\{-1,1\}$. The constraint $x^2 + y^2 + z^{2} \leq 1$ then ensures convergence of $x_t, y_t $ toward zeros, which shows that $\gamma_{t}$ converges a.s. toward $\{\rho_e,\rho_g\}$.

We now prove stabilization.  
Consider a flow function $\xi: [0,T] \to \mathcal{S}(\mathbb{H})$ 
and the corresponding dynamics
{\small\begin{align*}
\mathrm{d}\gamma_t^{\xi}
&= -\mathrm{i}\Big[\xi_z(t){\sigma}_z + \alpha_t\sigma_x,\gamma_t^{\xi}\Big]\mathrm{d}t
 + \big({\sigma}_z\gamma_t^{\xi}{\sigma}_z - \gamma_t^{\xi}\big)\mathrm{d}t\\
&\quad + \big({\sigma}_z\gamma_t^{\xi}+\gamma_t^{\xi}{\sigma}_z - 2\mathrm{tr}({\sigma}_z\gamma_t^{\xi})\gamma_t^{\xi}\big)\mathrm{d}W_t,
\end{align*}}
where $\xi_z(t) = \mathrm{tr}(\xi(t)\sigma_z)$. Consider the following Lyapunov function 
{\small\begin{align*}
    V(\gamma) = \sqrt{1 - \mathrm{tr}(\gamma\varrho^{\star})}.
\end{align*}}
One can check directly that conditions of \cite[Thereom 6.3]{liang19exponential}
are fulfilled. Hence, exponential stabilization is guaranteed for any admissible flow $\xi$. It remains to establish well-posedness of the mean-field controlled equation via a fixed-point argument. From the existence of the family of equations parametrized by $\xi$, we define the following mapping {\small $$\Xi:C\big([0,T], \mathcal{S}(\mathbb{H})\big)\to C\big([0,T], \mathcal{S}(\mathbb{H})\big)$$} by $\Xi(\xi):=(\mathbb{E}[\gamma_t^{\xi}])_{0 \leq t \leq T}$. Therefore the process $\gamma^m$ solves \eqref{eq:isingeq} if and only if $m = \Xi({m})$. To show that $\Xi$ admits a unique fixed point, observe that the coefficients of the SDE are locally Lipschitz (the feedback $\alpha^\star$ is linear).  
Using Itô's formula and standard estimates, for arbitrary elements $\xi^1,\xi^2$ we obtain
{\small
\begin{align*}
    \sup_{0 \leq r \leq t}\|\Xi(\xi^1)_r -  \Xi(\xi^2)_r\|_{2} \leq C\int_{0}^{t}\|\xi_s^{1} - \xi_{s}^{2}\|_{2}{\mathrm{d}s}, \quad \forall t \leq T.
\end{align*}}
Iterating the mapping yields, for any $k \ge 1$,
{\small
\begin{align*}\|\Xi^{(k)}(\xi^1)_t - \Xi^{(k)}(\xi^2)_t\|_{2} \leq  \frac{C^kt^{k-1}}{(k-1)!}\sup_{0 \leq r \leq t}\|\xi_s^{1} - \xi_{s}^{2}\|_{2}.
\end{align*}}
For $k$ sufficiently large, $\Xi^{(k)}$ becomes a contraction. By standard fixed-point arguments, $\Xi$ therefore admits a unique fixed point.
\end{proof}

\begin{remark}
The feedback law used above was first introduced in \cite{tsumura07} and later generalized in \cite{liang19exponential}. 
Its principle is to penalize deviations from the target state, while ensuring that the feedback vanishes once the target is reached, 
so that the desired state becomes an equilibrium of the closed-loop dynamics.
\end{remark}


Figure~\ref{fig:1-QSR} illustrates  quantum state reduction, while Figure ~\ref{fig:1-Stabilization} shows quantum state stabilization, both starting from the initial state ${\gamma}_{0} = \frac{1}{2}\mathbf{1}$.  

\begin{figure}[h!]
     \centering
  \includegraphics[width=\linewidth]{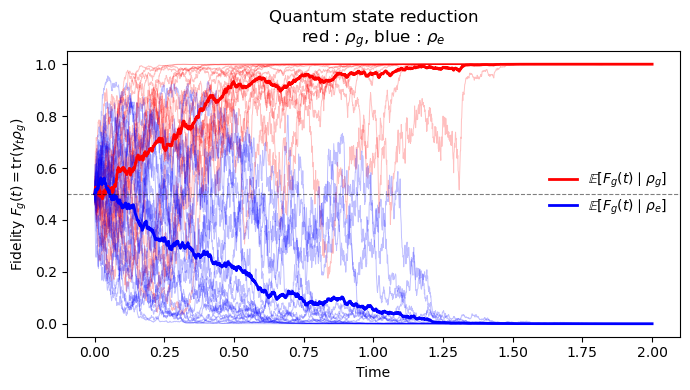}
  \caption{\footnotesize{Quantum state reduction: convergence of ${\gamma}_{t}$ toward $\{\rho_g,\rho_e\}$. }}
    \label{fig:1-QSR}
  \end{figure}

  \begin{figure}[h!]
     \centering
  \includegraphics[width=\linewidth]{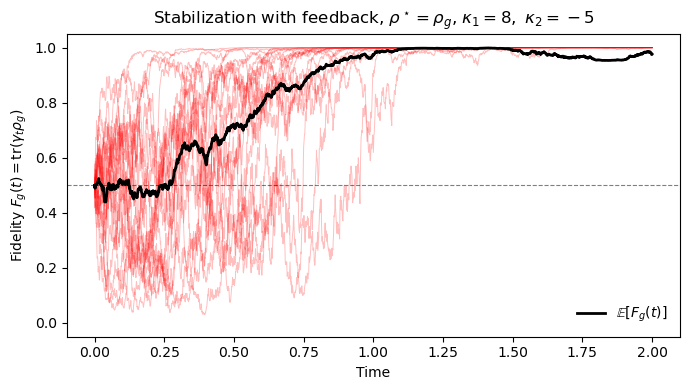}
\caption{\footnotesize{Stabilization of ${\gamma}_{t}$ toward ${\rho_{g}}$ when the control $\alpha^{\star}$ is applied.}}
    \label{fig:1-Stabilization}
  \end{figure}
}

\subsection{Two-player framework}
Returning to the case of a finite number of particles, we now take $N=2$ and illustrate a simple two-player quantum differential game.
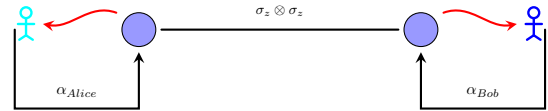
\begin{figure}
\begin{center}
\begin{tikzpicture}[scale=0.75, transform shape, every node/.style={scale=0.8}, >=stealth, line cap=round]

\definecolor{amethyst}{HTML}{9966CC}
\definecolor{amber}{HTML}{FFBF00}
\definecolor{aqua}{HTML}{00FFFF}

\tikzset{
stick/.pic={
\draw[pic actions, line width=0.9pt] (0,0) circle[radius=0.12];
\draw[pic actions, line width=0.9pt] (0,-0.12)--(0,-0.50);
\draw[pic actions, line width=0.9pt] (-0.18,-0.25)--(0.18,-0.25);
\draw[pic actions, line width=0.9pt] (-0.15,-0.60)--(0,-0.50)--(0.15,-0.60);
}
}

\draw[fill = blue!40] (0,-2) circle (0.3);
\draw[fill = blue!40] (5,-2) circle (0.3);

\draw[thick] (0.4,-2) -- (4.6,-2)
node[midway, above=4pt] {\small $\sigma_{z}\otimes\sigma_{z}$};

{\pic[draw=aqua] at (-2, -1.7) {stick};}
{\pic[draw=blue] at (7, -1.7) {stick};}

\draw[->, red, thick] (5.4,-1.7) to[out=20, in=190] (6.7,-1.9);
\draw[->, red, thick] (-0.4,-1.7) to[out=160, in=-10] (-1.7,-1.9);

\draw[<-, thick]
  (0,-2.4) -- ++(0,-1) -- ++(-2.2,0) node[midway, above=3pt] { $\alpha_{Alice}$} -- ++(0,1.4);

\draw[<-, thick]
  (5,-2.4) -- ++(0,-1) -- ++(2.2,0) node[midway, above=3pt] { $\alpha_{Bob}$} -- ++(0,1.4);

\end{tikzpicture} 

\caption{\footnotesize Two qubits interact via an Ising Hamiltonian $\sigma_z\otimes\sigma_z$.  
Alice and Bob each perform a continuous $z$-measurement and apply local feedback along the $x$-axis.}

\end{center}

\end{figure}
The system evolves according to
{\small\begin{align*}
\mathrm{d}\boldsymbol{\rho}_t
&= -\mathrm{i}[\tfrac{\sigma_z\otimes\sigma_z}{2}, \boldsymbol{\rho}_t]\mathrm{d}t
+ \sum_{l=1}^{2}\Big(\boldsymbol{\sigma}_{z,l}\boldsymbol{\rho}_t\boldsymbol{\sigma}_{z,l}-\boldsymbol{\rho}_t\Big)\mathrm{d}t \\
&\quad + \sum_{l=1}^{2}\Big(\boldsymbol{\sigma}_{z,l}\boldsymbol{\rho}_t+\boldsymbol{\rho}_t\boldsymbol{\sigma}_{z,l}
 -2\mathrm{tr}(\boldsymbol{\sigma}_{z,l}\boldsymbol{\rho}_t)\boldsymbol{\rho}_t\Big)\mathrm{d}W_t^{l},
\end{align*}}
with initial state $\boldsymbol{\rho}_0\in\mathcal{S}(\mathbb{C}^4)$.
The set of equilibrium points reduces to {$\mathfrak{F}\!=\{\rho_{gg},\,\rho_{eg},\,\rho_{ge},\,\rho_{ee}\}$}. 

\subsubsection{Feedback design for quantum 
state stabilization}
We now introduce local feedback controls aimed at stabilizing the joint two-qubit system toward the target state $\rho_{ge}$.  
Each agent applies a local control Hamiltonian along the $x$-axis:
{\small\begin{align*}
\widehat{\mathbf{H}}_{1,t} &= \alpha_{1}(\boldsymbol{\rho}_t)\boldsymbol{\sigma}_{x,1}, \qquad
\widehat{\mathbf{H}}_{2,t} = \alpha_{2}(\boldsymbol{\rho}_t)\boldsymbol{\sigma}_{x,2},
\end{align*}}
where the feedback laws are given by
{\small\begin{align*}
\alpha_{t,A} 
&= -\kappa_{1,A}\mathrm{i}\mathrm{tr}\Big([\boldsymbol{\sigma}_{x,1},\boldsymbol{\rho}](\rho_g\!\otimes\!\mathbf{1})\Big)
+\kappa_{2,A}\left(1-\mathrm{tr}(\boldsymbol{\rho}(\rho_g\!\otimes\!\mathbf{1}))\right),\\
\alpha_{t,B}
&= -\kappa_{1,B}\mathrm{i}\mathrm{tr}\Big([\boldsymbol{\sigma}_{x,2},\boldsymbol{\rho}](\mathbf{1}\otimes\rho_e)\Big)
+\kappa_{2,B}\!\left(1-\mathrm{tr}(\boldsymbol{\rho}(\mathbf{1}\!\otimes\!\rho_e))\right).
\end{align*}}

Under these feedback laws, the closed-loop system converges almost surely to the desired equilibrium state:
{$\boldsymbol{\rho}_t \longrightarrow \rho_{ge}$, as $t\to\infty.$}

\begin{remark}
The interpretation of the control strategy is as follows. The second term in each feedback law acts as a local stabilizer. All the points in $\mathfrak{F}$ are equilibria of the system. Hence, when Alice's qubit (resp. Bob's qubit) is close to an eigenstate that is not the regulation point, $\boldsymbol{\rho}_t$ must be prevented from converging to it. This is achieved by the first term.
In this setup, each agent only needs access to their own \emph{local state} to implement the feedback control.  
Let $\rho_t^{A}$ and $\rho_t^{B}$ denote the reduced density matrices obtained by partial tracing the global state $\boldsymbol{\rho}_t$ over the other subsystem:
{\small\begin{align*}
\rho_t^{A} = \mathrm{tr}_{B}(\boldsymbol{\rho}_t), \qquad 
\rho_t^{B} = \mathrm{tr}_{A}(\boldsymbol{\rho}_t).
\end{align*}}
Alice can therefore compute a feedback control based solely on her local estimate, for instance
{\small\begin{align*}
\alpha_1(t) := \alpha_1(\boldsymbol{\rho}_t^{A}), \qquad 
\boldsymbol{\rho}_t^{A} = \rho_t^{A}\otimes \tfrac{\mathbf{1}}{2},
\end{align*}}
such that we can write the control feedback strategy using only the reduced state: 
{\small \begin{align*}
\alpha_{1}(\rho^A) &= -\kappa_{1,A}\mathrm{i}\mathrm{tr}([\sigma_x,\rho^A]\rho_g) + \kappa_{2,A}\big(1-\mathrm{tr}(\rho^A \rho_g)\big),\\
\alpha_{2}(\rho^B) &= -\kappa_{1,B}\mathrm{i}\mathrm{tr}([\sigma_x,\rho^B]\rho_e) + \kappa_{1,B}\big(1-\mathrm{tr}(\rho^B\rho_e)\big).
\end{align*}}
\end{remark}

We illustrate both quantum state reduction and stabilization starting from $\boldsymbol{\rho}_{0}=\tfrac{1}{4}\mathbf{1}$.  
Figure~\ref{fig:2-QSR} shows a sample trajectory exhibiting state reduction, while Figure~\ref{fig:2-Stabilization} shows stabilization toward the target under the  controls $(\alpha_1,\alpha_2)$.

\begin{figure}[h!]
     \centering
  \includegraphics[width=\linewidth]{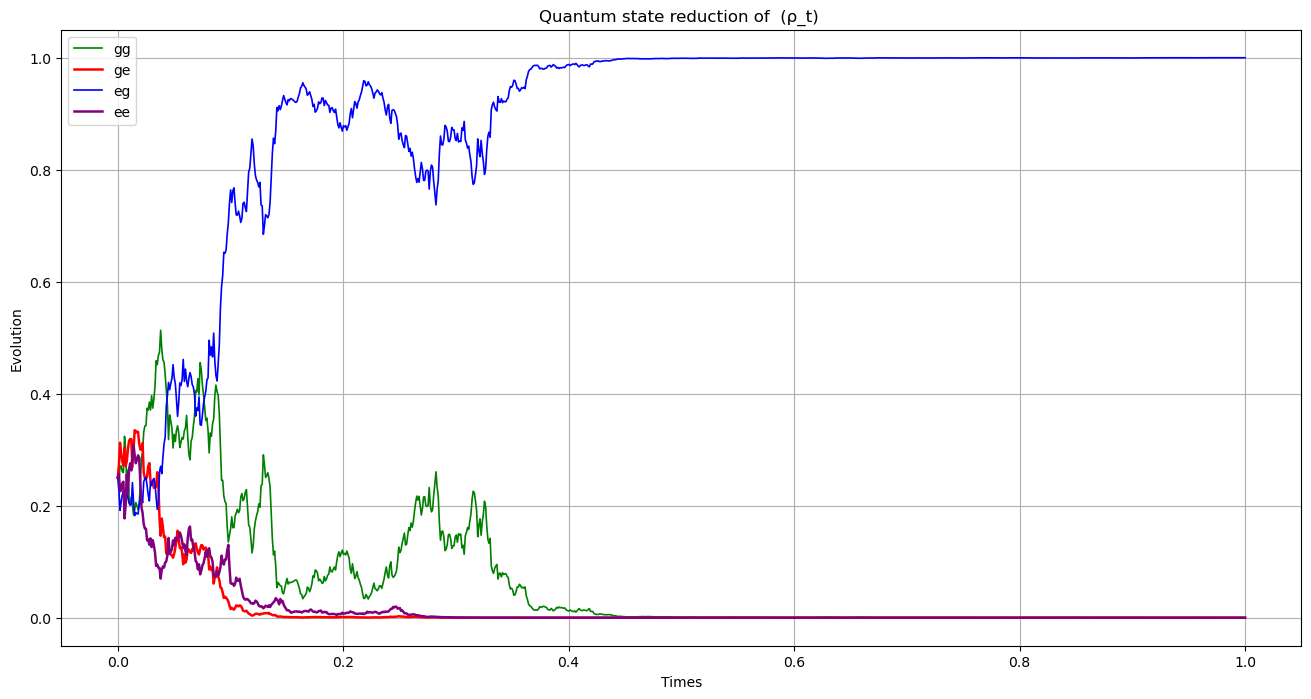}
  \caption{\footnotesize{In one realization, quantum state reduction: convergence of $\boldsymbol{\rho}_{t}$ toward $\textcolor{blue}{\rho_{eg}}$ }}
    \label{fig:2-QSR}
  \end{figure}

  \begin{figure}[h!]
     \centering
  \includegraphics[width=\linewidth]{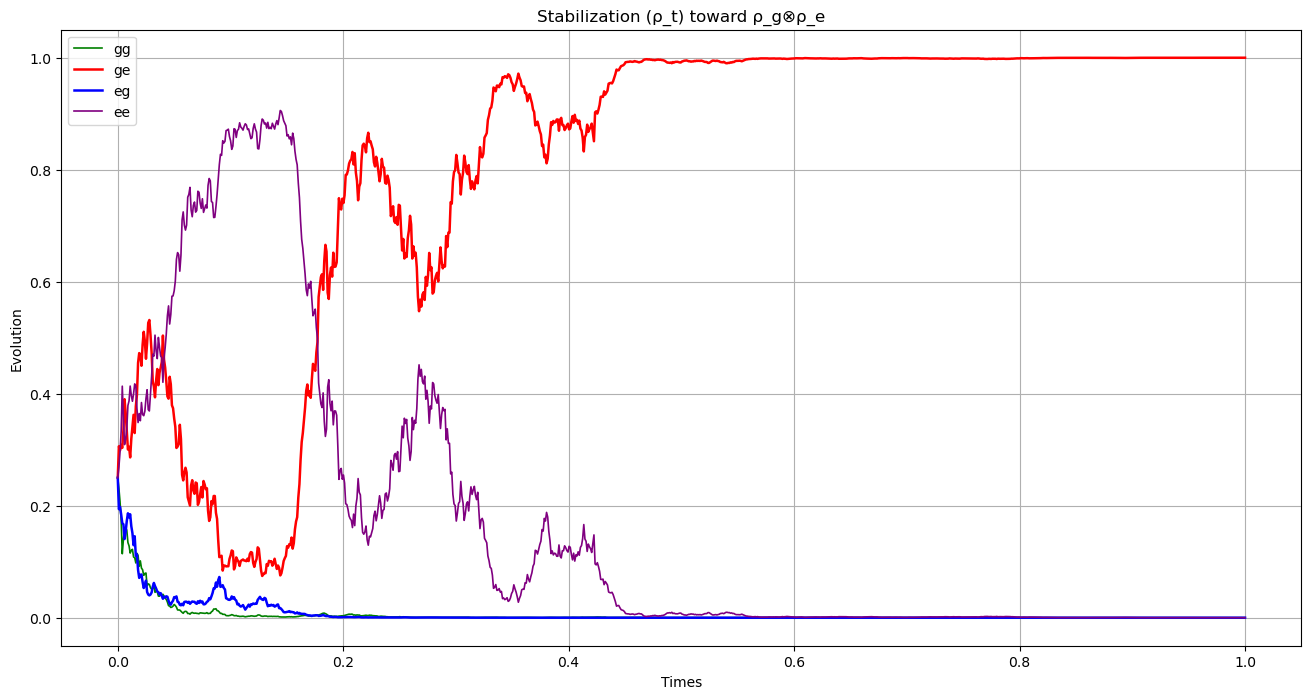}
  \caption{\footnotesize{In one realization, stabilization  of $\boldsymbol{\rho}_{t}$ toward $\textcolor{red}{\rho_{ge}}$}, when each player adopt control strategy $(\alpha_1,\alpha_2)$}
    \label{fig:2-Stabilization}
  \end{figure}

\subsubsection{Game-theoretic formulation}

To formalize the control scenario as a game, we specify each player's objective via a cost functional. Each player $l\in\{1,2\}$ seeks to minimize
{\small\begin{align*}
\mathcal{J}_l(\boldsymbol{\rho},\alpha_{l})
&= \mathbb{E}\left[\int_{0}^{T}\mathrm{C}_l\big(\boldsymbol{\rho}_{t},\alpha_{t,l}\big)\mathrm{d}t + \mathrm{F}_l(\boldsymbol{\rho}_T)\right].
\end{align*}}

Since the goals of the two players do not conflict in this setting, the two-player game is effectively cooperative, in the sense that they may be assigned the same value function to minimize:
{\small
\begin{align*}
\mathcal{J}(\boldsymbol{\rho}, \alpha_{l})
&\!:=\! \mathbb{E}\!\left[\!\int_{0}^{T}\Big(
1\!-\!\mathrm{tr}\big(\boldsymbol{\rho}_{t}\rho_{ge}\big)
+\!\frac{\alpha_{t,l}^{2}}{2}\!
\Big)\mathrm{d}t
+\Big(1\!-\!\mathrm{tr}\big(\boldsymbol{\rho}_{T}\rho_{ge}\big)\Big)\!\right].
\end{align*}
}

In another scenario, if Alice aims for $\rho_{gg}$ while Bob aims for $\rho_{ee}$, the two local controllers bias the collapse in \emph{opposite} directions. In this case the closed loop becomes \emph{bistable}: trajectories tend to concentrate near $\rho_{ge}$, and no control profile can simultaneously minimize both agents’ global costs (Alice prefers $\rho_{gg}$, Bob prefers $\rho_{ee}$). This corresponds to a competitive game, where each player controls only their own qubit and $\rho_{ge}$ emerges as the compromise outcome.

\section{CONCLUSIONS}
\label{sec:conc}
Future work should focus on developing computational methods for differential quantum games and quantum mean-field games. Since numerical implementations necessarily rely on discrete approximations, it would also be useful to formulate a discrete-time framework based on repeated measurement setups.

\addtolength{\textheight}{-12cm}   

\bibliographystyle{plain}
\bibliography{Bibliography}

\end{document}